\documentclass[letterpaper, 10 pt, conference]{ieeeconf}  
\IEEEoverridecommandlockouts             
\usepackage{amssymb} %

\usepackage{bm}
\usepackage{enumerate}
\usepackage{commath}
\usepackage{graphicx}
\graphicspath{{Pictures/}}
\usepackage{amsmath}
\usepackage{mathtools}
\usepackage{subcaption}
\usepackage{breqn}
\usepackage{cite}
\usepackage[table]{xcolor}
\usepackage{booktabs}
\usepackage{empheq}
\usepackage{amsfonts}

\usepackage{amsthm}
\usepackage{hyperref}
\usepackage{xcolor}
\usepackage{mathrsfs}
\usepackage{accents}
\usepackage{thmtools}
\usepackage{thm-restate}
\usepackage{float}
\usepackage{xcolor}
\usepackage[normalem]{ulem}

\usepackage{algorithm}
\usepackage{algpseudocode}
\usepackage{dblfloatfix}

\theoremstyle{plain}
\newtheorem{corollary}{Corollary}
\newtheorem{lemma}{Lemma}

\newtheorem{proposition}{Proposition}

\newtheorem*{theorem*}{Theorem}
\newtheorem{assumption*}{Assumption}

\declaretheorem[name=Theorem]{thm}

\theoremstyle{definition}

\newtheorem{remark}{Remark}
\newtheorem{definition}{Definition}
\newtheorem{assumption}{Assumption}
\newtheorem{problem}{Problem}
\newtheorem*{problem*}{Problem}

\newcommand{\myset}[1]{\mathcal{#1}}

\DeclareMathOperator*{\argmin}{arg\,min}

\title{ \LARGE \bf
  Safety of Linear Systems under Severe Sensor Attacks
}

\author{Xiao Tan, Pio Ong, Paulo Tabuada,  and Aaron D. Ames%
\thanks{This work is supported by TII under project \#A6847.}
\thanks{Xiao Tan, Pio Ong, and Aaron D. Ames are with the the Department of Mechanical and Civil Engineering, California Institute of Technology, Pasadena, CA 91125, USA (Email: {\tt\small xiaotan, pioong, ames@caltech.edu}).}    
\thanks{Paulo Tabuada is with the Department of Electrical and Computer Engineering at University of California, Los Angeles, CA 90095, USA (Email: {\tt\small  tabuada@ucla.edu}).} 
}

\begin{document}

\maketitle
\thispagestyle{plain} 
\pagestyle{plain}

\begin{abstract}
Cyber-physical systems can be subject to sensor attacks, e.g., sensor spoofing, leading to unsafe behaviors. 
This paper addresses this problem in the context of linear systems when an omniscient attacker can spoof several system sensors at will.  
In this adversarial environment, existing results have derived necessary and sufficient conditions under which the state estimation problem has a unique solution. In this work, we consider a severe attacking scenario when such conditions do not hold. To deal with potential state estimation uncertainty, we derive an exact characterization of the set of all possible state estimates. Using the framework of control barrier functions, we propose design principles for system safety in offline and online phases. For the offline phase, we derive conditions on safe sets for all possible sensor attacks that may be encountered during system deployment.  For the online phase, with past system measurements collected, a quadratic program-based safety filter is proposed to enforce system safety. A $2$D-vehicle example is used to illustrate the theoretical results.
  \end{abstract}

\section{Introduction}

Cyber-physical systems (CPS), as an integration of computation, communication, and physical processes, range from small-scale applications, such as cars, to large-scale infrastructure, such as smart grids and water distribution systems. Because of their close interaction of physical and computational components, CPS are prone to attacks in both cyber and physical domains. Prominent examples of such attacks include the Stuxnet malware\cite{langner2011stuxnet} targeting on a process control system. Previous works also demonstrate that attacks may come from the physical domain \cite{shoukry2013non,tu2018injected}. In \cite{shoukry2013non}, the authors spoof the velocity measurements of a vehicle, which, when intervened by anti-lock braking systems, may cause the drivers to lose control of their vehicles. Motivated by these real-world examples, researchers have explored various attacking and defense strategies for CPS, including, e.g., denial-of-service \cite{de2015input}, replay attacks \cite{zhu2013performance}, man-in-the-middle \cite{smith2015covert}, and false data injection \cite{mo2010false}. 

In this paper, we consider a scenario similar to those in \cite{shoukry2013non,tu2018injected} wherein certain measurements of the CPS are compromised by an attacker. In the setting of linear systems, we adopt a general attack model that imposes no limitations on the magnitude, statistical properties, or temporal evolution requirements on the attack signal. Rather, we only assume an upper bound on the number of attacked sensors. Most existing results in this setting have focused on recovering the system state from compromised measurement data, known as the secure state reconstruction problem. For discrete-time linear systems, \cite{fawzi2014secure} derives several necessary and sufficient conditions for the uniqueness of solutions to this problem, which are further refined in \cite{shoukry2015event}. There, the condition is posed as a sparse observability property of the CPS. An equivalent condition for continuous-time linear systems is given in \cite{chong2015observability}.  Recently, \cite{mao2022computational} shows that finding this unique solution is NP-hard in general. 

An important question naturally arises for CPS in this adversarial scenario: can we ensure \emph{safety} of the system, and thereby avoid catastrophic results through active control, even when certain sensors are compromised? Safety in control systems usually refers to the property that system trajectories can be made to stay within a safe set via feedback. Compromised sensor measurements will negatively affect or mislead the state estimates, and thus complicate safe control design. Motivated by these challenges, \cite{lin2023secondary} and \cite{zhang2022safe} have explored safe control designs in this setting. In~\cite{lin2023secondary}, the sensor attack signals are taken to be bounded, and a set of safeguarded sensors are assumed available for the design of the so-called secondary controller. In \cite{zhang2022safe}, a finite attack pattern is assumed, and each pattern corresponds to a particular subset of compromised sensors. A fault identification scheme is proposed by implementing a large number of extended Kalman filters simultaneously. They further make assumptions on the sensor attacks so that the state estimate error is bounded in probability. The assumptions made by these works may not hold true under the aforementioned sensor attack model that we adopt.

In recent years, control barrier functions (CBF) \cite{Ames2017} have gained popularity as a framework for safety-critical control---providing Lyapunov-like necessary and sufficient conditions for forward set invariance. 
A strength of this approach is its ability to incorporate (bounded) uncertainty, including uncertainty in the input \cite{kolathaya2018input}, 
uncertainty in the state \cite{agrawal2022safe}, 
and measurement \cite{dean2021guaranteeing}. 
There have also been non-deterministic characterizations of safety in the context of risk-adverse CBFs \cite{singletary2022safe,vahs2023belief}, along with stochastic CBFs \cite{Cosner-RSS-23}, where the possible state/perturbation follows a certain stochastic distribution. 
Yet when unbounded, non-stochastic, intelligent  sensor attacks are performed by an omniscient attacker, the state estimation error does not satisfy the assumptions of these works. Note that CBFs have been applied to address other security issues, such as privacy preservation  \cite{zhong2023secure} and safety in the presence of faulty sensors\cite{zhang2022safe}.

In this work, we focus on safety guarantees for CPS subject to general sensor attacks described above. In particular, we consider scenarios where the solution to the secure state reconstruction problem may not be unique. Our contributions are summarized as follow:
\begin{enumerate}
    \item  We provide an exact characterization of the set of possible solutions to the secure state reconstruction problem for linear discrete-time systems.
    \item We outline design principles for safe sets in the offline phase for the worst-case attacking scenario under a mild sparse observability assumption.    
    \item We propose an online safe control scheme that provides safety guarantees in the presence of possibly unbounded state estimation error.
\end{enumerate}
Ultimately, this paper presents a general characterization of the safety of linear systems subject to \emph{severe sensor attacks}. 

\textbf{Notation:} For $ w\in \mathbb{N}$, define $[w]:=\{1,2,...,w\}$. The cardinality of a set $\myset{I}$ is denoted by $|\myset{I}|$. Given a $w\in \mathbb{N}$, a $k$-combination from $[w]$ is a subset of $[w]$ with cardinality $k$. Denote by $\mathbb{C}_{w}^k$ the set of all $k$-combinations from $[w]$. For a matrix $C\in \mathbb{R}^{w\times n}$ and an index set $\Gamma\subseteq[w]$, denote by $C_{\Gamma}$ the matrix obtained from $C$ by removing all the rows with indices not in $\Gamma$. For a point $x\in \mathbb{R}^n$, a set $\myset{X}\subseteq \mathbb{R}^n$, and a matrix $A\in \mathbb{R}^{n\times n}$, we define $\| x\|_{\myset{X}} := \min_v \| x-v\| \textup{ s. t. } v\in \myset{X}$,  and $A(\myset{X}) := \{y\in \mathbb{R}^{n}: y = Ax, \, x\in \myset{X} \} $. Minkowski summation $\myset{X}_1 + \myset{X}_2 $  for two sets $\myset{X}_1, \myset{X}_2$ is defined as $ \{x_1 + x_2, x_1\in \myset{X}_1, x_2 \in \myset{X}_2\}$. When no confusion arises, with slight abuse of notation, we use interchangeably a vector $x\in \mathbb{R}^n$ and the singleton set $\{x\}$.

\section{Problem formulation}

Consider a discrete-time linear system under sensor attacks
\begin{equation} \label{eq:system}
		\begin{aligned}
		\text{(dynamics)} \hspace{0.5cm} & x(\tau+1) = Ax(\tau) + Bu(\tau),\\
			\text{(measurement)} \hspace{0.5cm} & y(\tau)  = Cx(\tau) + e(\tau), \\
			\text{(safe set)} \hspace{0.5cm}   \myset{C} = & \{x\in \mathbb{R}^n: h(x):= Hx + q \geq 0\},
	\end{aligned}
	\end{equation}
where  $x(\tau) \in \mathbb{R}^n,$ $ u(\tau) \in \mathbb{R}^m,$  $ y(\tau)\in \mathbb{R}^p,$ $e(\tau)\in \mathbb{R}^p$, represent the system state, the control input,  the output measurement, and the attacking signal on the sensors, respectively. Regarding the attacks, $e_i(\tau)$ is nonzero whenever a sensor $i\in [p]$ is under attack at time $\tau$.   
A safe set $\myset{C}$ is the set of states that are permitted along the system trajectories.  Here, we specify the set with a vector-valued linear function $ h: \mathbb{R}^n \to \mathbb{R}^l$, representing a polytopic safe set. For notation simplicity, we refer to the input-output data $ (u(0),u(1), \dots, u(t-1),  y(0),y(1), \dots, y(t))\in \mathbb{R}^{mt + p(t+1)}$ collected after receiving the latest measurement and before choosing a control input  at time $t$ as $\mathcal{D}_t$.

We adopt the sensor attack model presented in \cite{fawzi2014secure,shoukry2015event}, and assume the following throughout the paper. 

\begin{assumption}\label{ass:attacking model}
    The attacker has full knowledge of the system, including the states, the dynamics~\eqref{eq:system}, and our defense strategy. It may choose $s$ out of the $p$ sensors to attack, the choice of which remains unchanged for the duration considered. For these sensors, the attacker can adjust the corresponding $e_i(\tau)$ to any value.
\end{assumption}

{This attack model covers a wide range of sensor anomalies, including sensor failures, sensor faults \cite{zhang2021dynamic}, and sensor attacks. During an attack, the attacker intentionally modifies measurements from some sensors to pursue specific objectives, which may lead to the worst-case sensor anomaly in terms of system safety.
Assumption~\ref{ass:attacking model} only restricts the total number of sensor anomalies up to $s$.}

We now introduce the notions of system safety and control barrier functions. Note that we restrict the class of barrier functions to be linear.

\begin{definition}[System safety]
    A system can be rendered safe from time $t$ if there exists an input sequence $\{u(\tau)\}_{\tau\geq t}$ 
    such that 
     the state $x(\tau) $ remains in the safe set $\myset{C}$ for all $\tau  \geq t$.
\end{definition}
\begin{definition}[Control barrier function\cite{agrawal2017discrete}]
    The function $h: x\mapsto Hx + q$ with $H\in \mathbb{R}^{l\times n}$ and $ q \in \mathbb{R}^l$ is called a \emph{(discrete-time) control barrier function} for system \eqref{eq:system} if, for some $\gamma\in (0,1)$, for each $x\in \mathbb{R}^n$, there exists an input $u\in \mathbb{R}^m$ such that 
    \begin{equation} \label{eq:cbf_condition}
        H(Ax + Bu) + q  \geq (1-\gamma) (Hx + q). 
    \end{equation}
\end{definition}
It has been shown in \cite{Ames2017,agrawal2017discrete} that if   $u(\tau)$ is chosen to satisfy the condition in \eqref{eq:cbf_condition} along system trajectories, then the safe set $\myset{C}$ is forward invariant and asymptotically stable. Hereafter we refer to \eqref{eq:cbf_condition} as the CBF condition.

When system sensors are under attack, one immediate question is ``can a good state estimate be retained?" {Since our system is known,  the state estimation problem can be reduced to identification of all plausible initial states.} Related to this are the following notions. 
\begin{definition}[$r$-sparse observability]
    System \eqref{eq:system} is \emph{$r$-sparse observable} if the pair $(A, C_{\Gamma})$ is observable for any index set $\Gamma\in \mathbb{C}_{p}^{p-r}$.
\end{definition}

\begin{definition}[Plausible initial states] \label{def:plausible state}
     Given input-output data $\mathcal{D}_t$, we call  $x_0$ \emph{a plausible initial state} if there exists $e(\tau), \tau 
 = 0,1,...,t,$ satisfying Assumption \ref{ass:attacking model} such that the first two  equations in \eqref{eq:system} hold with $x(0) = x_0$.  We denote the set of all plausible initial states at time $t$ by $\myset{X}_t^0$. 
 \end{definition}

It is known that there exists a unique solution to the secure state reconstruction problem, i.e., the set of all plausible initial state $\myset{X}_t^0$ is a singleton, if and only if the system is $2s$-sparse observable \cite{shoukry2015event}. In this work, we investigate the case beyond the $2s$-sparse observability condition, hereafter referred to as the  \emph{severe sensor attacks case}. In this case, the plausible initial states may not be unique. As will be shown later, this ambiguity complicates system safety analysis. 

In this work, we consider system safety guarantees under severe sensor attacks both in offline and online phases. For the offline phase, we characterize the set of  plausible states under the worst-case sensor attack, and identify sets in the state space for which the system can be rendered safe. For the online phase, where measurements are collected for the first $t$ steps, we derive control barrier function-based conditions that guarantee system safety from time $t$. This is formally stated as follows.

\begin{figure*}[!b]
    \normalsize
\begin{equation} \label{eq:matrices} \tag{5}
\widetilde{Y}_i = \begin{bsmallmatrix}
        y_i(0) \\
        y_i(1) \\
        \vdots \\
        y_i(t)
    \end{bsmallmatrix},~
    E_i = \begin{bsmallmatrix}
        e_i(0) \\
        e_i(1) \\
        \vdots \\
        e_i(t)
    \end{bsmallmatrix},~
    U = \begin{bsmallmatrix}
        u(0) \\
        u(1) \\
        \vdots\\
        u(t-1) \\
        0
    \end{bsmallmatrix},~
    \mathcal{O}_i = \begin{bsmallmatrix}
        C_i \\
        C_i A \\
        \vdots \\
        C_i A^{t}
    \end{bsmallmatrix},~
    F_i = \begin{bsmallmatrix}
        0 & 0 & \cdots &  0 & 0 \\
        C_iB & 0 & \cdots & 0 & 0  \\
        \vdots & & \ddots & & 0 \\
        C_iA^{t-1}B &  C_iA^{t-2}B & \cdots & C_i B & 0
    \end{bsmallmatrix}
\end{equation}
\end{figure*}
\begin{problem}[Worst-case sensor attack]
    Given system matrices $A, B, C$, and the number of sensor attacks $s$, derive conditions on $H$ and $ q$ such that the system can be rendered safe under all possible sensor attacks.
\end{problem}

\begin{problem}[Fixed yet unknown sensor attack]
    Given system matrices $A, B, C$, the number of sensor attacks  $s$, and the input-output data $\mathcal{D}_t $, $t\geq n-1$, derive conditions on $H, q,$ and the input sequence $ \{u(\tau)\}_{\tau \geq t}$ so that the system is safe. 
\end{problem}

 \section{Secure state estimation and safe control}
 In this section, we present our main results on system safety under sensor attacks. We approach this problem in two steps: 1) secure state estimation and 2) safe control design. First, we characterize the set of all plausible states mathematically. We then derive conditions that guarantee all state trajectories emanating from the set of plausible initial states  $\myset{X}_t^0$ remain within the safe set. We highlight that the set $\myset{X}_t^0$ may contain multiple, possibly infinite elements.

 \subsection{Characterization of plausible states}

 In this section, we only consider the case $t\geq n-1$. When $t<n-1$, some states may become plausible states due to the lack of measurements, which is not the focus of this work.

 Following the derivations in \cite{shoukry2015event}, by stacking the measurement history per sensor, we can rewrite the equations in \eqref{eq:system} in a compact form as follows: 
\begin{equation}\label{eq:measurement_matrix_eq}
    Y_i = \mathcal{O}_i x_0 + E_i,~i = 1,2,...,p,
\end{equation}
 where the matrices are defined in \eqref{eq:matrices} at the bottom of this page, and $Y_i = \widetilde{Y}_i - F_i U$. We establish the following result on the set $\myset{X}_t^0$.
\begin{proposition} \label{prop:state uncertainty X1} 
Given system \eqref{eq:system}, the maximal number of attacked sensors $s$, and the input-output data $ \mathcal{D}_t$, we have:
    \begin{equation} \label{eq:X0set}
        \myset{X}_t^0 = \bigcup_{\forall \  \{i_1, ..., i_{p-s}\} \in \mathbb{C}_{p}^{p-s}} 
 \left( \myset{X}_t^{0,i_1} \cap \myset{X}_t^{0,i_2}\cap ... \cap \myset{X}_t^{0,i_{p-s}} \right)
    \end{equation}
  where   $\myset{X}_t^{0,i} = \{x\in \mathbb{R}^n: \mathcal{O}_i x = Y_i\}$ for $i \in [p]$.
\end{proposition}
\begin{proof}
    First we show every element {from the right hand side of \eqref{eq:X0set}} is a plausible initial state under $s$ sensor attacks. Consider any $x_0$ {from the union}. The state $x_0$ must belong to one of the sets in the union. That is, there exists a set of $p-s$ indices $\Gamma =\{i_1, ..., i_{p-s}\}\in \mathbb{C}_{p}^{p-s}$ such that  $\mathcal{O}_i x_0 = Y_i$ for $i\in \Gamma$. Then, because $|\Gamma|=p-s $, there are precisely $s$  sensor indices left, $[p]\setminus\Gamma$. We can choose $E_j = Y_j - \mathcal{O}_j x,~\forall j\in [p]\setminus\Gamma$ so that \eqref{eq:measurement_matrix_eq} holds for each index $j$. Thus, there exists an attacking signal for which $E_i=0$ for $p-s$ indices and nonzero for $s$ indices, thereby satisfying Assumption \ref{ass:attacking model}, such that  \eqref{eq:measurement_matrix_eq} holds, and hence $x_0$ is a plausible initial state. {That is, $x_0 \in \myset{X}_t^0$}.

    Now we show that no state other than those in {the union} in \eqref{eq:X0set} is a plausible initial state with the measurements. Suppose there exists a plausible state {$x_0 \in \myset{X}_t^0$ but not from the union}. Then {because there can only be $s$ attacked sensors}, there must exist a set of indices $\Gamma\in \mathbb{C}_{p}^{p-s}$ where $E_i=0$. As a result, we have $\mathcal{O}_i x_0 = Y_i, i \in \Gamma$ from \eqref{eq:measurement_matrix_eq}. This, however, suggests that $x_0$ belongs to one of the sets in the union in \eqref{eq:X0set}, which is a contradiction.
\end{proof}

Following \eqref{eq:X0set}, one observes that, in general, the set of plausible initial states $\myset{X}_t^{0}$ is a union of affine subspaces. Based on this result, we further derive the set of plausible states at time $t$ by forward propagating $\myset{X}_t^{0}$ through the system~\eqref{eq:system}.
To ease the notation, for a combination $\Gamma = \{i_1, i_2, ..., i_{p-s}\}\in \mathbb{C}_{p}^{p-s}$, let $\myset{X}_{t}^{0,\Gamma} :=\myset{X}_t^{0,i_1} \cap \myset{X}_t^{0,i_2}\cap ... \cap \myset{X}_t^{0,i_{p-s}} = \{x: \mathcal{O}_{\Gamma}x = Y_{\Gamma} \}$, where 
\begin{equation} \label{eq:def_O_Gamma}
 \setcounter{equation}{\value{equation}+1}
    \mathcal{O}_{\Gamma} := \begin{bsmallmatrix}
        \mathcal{O}_{i_1} \\
        \mathcal{O}_{i_2} \\
        \vdots \\
        \mathcal{O}_{i_{p-s}}
    \end{bsmallmatrix},~Y_{\Gamma} := \begin{bsmallmatrix}
        Y_{i_1} \\
        Y_{i_2} \\
        \vdots \\
        Y_{i_{p-s}}
    \end{bsmallmatrix}.
\end{equation} 
Furthermore, we use the notation $\myset{X}^{\tau}_t$ and $\myset{X}^{\tau,\Gamma}_t$  for the sets of all plausible states at time $\tau$, starting from $\myset{X}^{0}_t$ and $\myset{X}^{0,\Gamma}_t$, respectively. We will omit the subscript $t$ when it is clear the measurements considered are from the first $t$ steps. We now characterize the set $\myset{X}^{t}_t$.
\begin{corollary} \label{col:X_t set}
Under the premises of Proposition \ref{prop:state uncertainty X1}, we have 
\begin{equation}
     \myset{X}^{t} = \bigcup_{\forall \Gamma \in \mathbb{C}_{p}^{p-s}}   \myset{X}^{t,\Gamma} 
\end{equation}
    where, for each $\Gamma\in\mathbb{C}_{p}^{p-s}$,
\begin{multline}
    \label{eq:X_t_Gamma}
    \myset{X}^{t,\Gamma} = A^t( \myset{X}^{0,\Gamma}) + A^{t-1}Bu(0) +  ... + Bu(t-1).
\end{multline}

\end{corollary}
\begin{proof}
     The result follows from a direct calculation of the system dynamics 
$
    \myset{X}_t^{\tau+1,\Gamma} = A (\myset{X}_t^{\tau,\Gamma}) + B u(\tau)
$ for all $\tau \in \{0,1,...,t-1\}$.
\end{proof}

\subsection{Offline phase safe set design}
In this subsection, we establish  conditions on $(A, B, $ $C, s, H, q)$ that provide safety guarantees under all possible attacking scenarios. We start by  characterizing $\myset{X}^0$ for all possible input-output data generated by the system under a certain sparse observability condition.

\begin{proposition} \label{prop:general sparse case}
    Consider system \eqref{eq:system} with $s$ sensors under attack. Let $x_{\textup{true}}$ denote the true but unknown initial state. If the system is $s$-sparse observable, then for any attack signal $\{e(\tau)\}_{0\leq\tau\leq t}$ assigned by the attacker,
    the set of all plausible initial states $\myset{X}^0$ is a finite set. Moreover, when $p > 2s$, the set of plausible initial states $\myset{X}^0$ satisfies
    \begin{equation} \label{eq:general X0}
        \myset{X}^0 \subset \{x_{\textup{true}} \} + \bigcup_{\Lambda\in \mathbb{C}_{p}^{p-2s}} \ker(\mathcal{O}_{\Lambda}).
    \end{equation}
\end{proposition}

    \begin{proof}
    Since $(A,C)$ is $s$-sparse observable,
    $\mathcal{O}_\Gamma$ is full column rank for any $\Gamma \in \mathbb{C}_p^{p-s}$. Therefore, each set $\myset{X}^{0,\Gamma}$
    is either a singleton or an empty set. From \eqref{eq:X0set}, $\myset{X}^0 = \bigcup_{\Gamma \in \mathbb{C}_p^{p-s}} \myset{X}^{0,\Gamma}$ is a union of finite sets, so it is finite.

Now consider the case $p> 2s$.  Because we can write $x = x_{\textup{true}} +(x-x_{\textup{true}})$, we will show that $x-x_{\textup{true}}\in \bigcup_{\Lambda\in \mathbb{C}_{p}^{p-2s}} \ker(\mathcal{O}_{\Lambda})$ for any $x\in \myset{X}^0$. Since $\myset{X}^{0,\Gamma}$ is either an empty set or a singleton, there exist two distinct combinations $\Gamma_0,\Gamma\in\mathbb{C}_{p}^{p-s}$ with  $\mathcal{O}_{\Gamma_0} x_{\textup{true}} = Y_{\Gamma_0} $ and $\mathcal{O}_{\Gamma} x = Y_{\Gamma}$. We note the following for $\Gamma_0, \Gamma$ from $\mathbb{C}_{p}^{p-s}$:
\begin{equation}\label{eq:combination_bounds}
   p-2s\leq | \Gamma_0 \cap \Gamma| \leq p-s-1
\end{equation} 
because the two sets each have $p-s$ elements out of $p$ total elements, and they must also be distinct. The lower bound implies the intersection $\Gamma_0\cap \Gamma \neq \emptyset$ is nonempty, so $\mathcal{O}_{\Gamma_0\cap\Gamma}$ and $Y_{\Gamma_0\cap\Gamma}$ are well-defined. Then, the equations
$
\mathcal{O}_{\Gamma_0\cap\Gamma}x_{\textup{true}} = Y_{\Gamma_0\cap\Gamma},$ and $ \mathcal{O}_{\Gamma_0\cap\Gamma}x = Y_{\Gamma_0\cap\Gamma}$ are nontrivial.
Subtracting these two equations yields:
\begin{equation}
    x - x_{\textup{true}} \in \ker(\mathcal{O}_{\Gamma_0 \cap \Gamma})\subseteq \bigcup_{\Gamma_1,\Gamma_2\in \mathbb{C}_{p}^{p-s}}\ker(\mathcal{O}_{\Gamma_1 \cap \Gamma_2}).
\end{equation}
From \eqref{eq:combination_bounds},
we only need to consider $\Lambda = \Gamma_1\cap \Gamma_2$ such that $p-2s\leq|\Lambda|\leq p-(s+1)$:
\begin{equation*}
     x - x_{\textup{true}} \in \bigcup_{k = p-2s}^{p-(s+1)} \left(\bigcup_{\Lambda\in \mathbb{C}_{p}^{k}} \ker(\mathcal{O}_{\Lambda})\right).
\end{equation*}

One can verify that 
for any combinations $\bar{\Lambda}, \Lambda \subseteq [p]$, if $\Lambda \subseteq \bar{\Lambda}$, then $    \ker(\mathcal{O}_{\bar{\Lambda}})\subseteq \ker(\mathcal{O}_{\Lambda}) $. Thus, we know $\bigcup_{\Lambda\in \mathbb{C}_{p}^{p-(s+1)}} \ker(\mathcal{O}_{\Lambda}) \subseteq  ... \subseteq \bigcup_{\Lambda\in \mathbb{C}_{p}^{p-2s}} \ker(\mathcal{O}_{\Lambda})$. This further implies that \eqref{eq:general X0} holds. \end{proof}
Proposition \ref{prop:general sparse case} characterizes the set~$\myset{X}^0$ when there are sensor redundancies. The result shows that even under severe sensor attacks, $s$-sparse observability still ensures a finite number of plausible states, and with enough sensors, the attackers can only confuse the system with plausible initial states contained within the set given by~\eqref{eq:general X0}. We note that the existing result on $2s$-sparse observability follows directly from our characterization.
\begin{corollary}[\cite{shoukry2015event}]
    When system \eqref{eq:system} is $2s$-sparse observable,
    we have $\myset{X}^0 = \{x_{\textup{true}}\}$. 
\end{corollary}
\begin{proof}
    $2s$-sparse observability implies $p>2s$.   The result follows from~\eqref{eq:general X0} since $ \ker(\mathcal{O}_{\Lambda}) = \{0\},~\forall \Lambda\in \mathbb{C}_{p}^{p-2s}$.
\end{proof}
\begin{corollary}
    When system \eqref{eq:system} is not $s$-sparse observable, there exists some possible attack signal $\{e(\tau)\}_{0\leq\tau\leq t}$ assigned by the attacker such that $\myset{X}^0$ is an infinite set.
\end{corollary}
\begin{proof}
    Since $(A,C)$ is not $s$-sparse observable, we know $\mathcal{O}_\Gamma$ is not full column rank for some $\Gamma \in \mathbb{C}_p^{p-s}$. Denote one such index set by $\Gamma_0$. Suppose that sensors in the index set $\Gamma_0$ are intact, then $\myset{X}^{0,\Gamma_0}\subset \mathbb{R}^n$ is an affine subspace. From \eqref{eq:X0set}, $\myset{X}^0 $ is thus an infinite set.
\end{proof}

We now restrict our discussion to the case when the system is $s$-sparse observable {and derive conditions on $H$ and $q$ for system safety for all possible attacks}. When this condition fails, the sensor attack is undetectable\footnote{By attack detection, we mean that whether, for any attacks on $s$ sensors, the system can tell apart the case where every sensors are intact and the case where there are sensor attacks in the system. An attack detection  on $s$ sensors is possible  if and only if the system is $s$-sparse observable. See \cite[Theorem 16.1]{diggavi2020coding}  for more details. }.

\begin{thm} \label{thm:offline design}
    Consider system \eqref{eq:system} with $s$ sensors under attack. Assume that the system is $s$-sparse observable and $p>2s$. If the following conditions hold:
    \begin{enumerate}
        \item[i.]  for any $\Lambda\in \mathbb{C}_{p}^{p-2s}$, $\ker(\mathcal{O}_{\Lambda}) \subseteq \ker\left(\begin{bsmallmatrix}
           H \\
           HA \\
           \vdots \\
           HA^{n-1}
       \end{bsmallmatrix}\right)$;
       \item[ii.] $h(x) = Hx + q$ is a control barrier function.
    \end{enumerate}
       Then there exist maps $k_t:\mathbb{R}^{mt + p(t+1)} \to \mathbb{R}^m$ such that the system under the feedback $u(t) = k_t(\mathcal{D}_t)$  satisfies the following implications: 
       \begin{equation} \label{eq:safety_guarantee}
           \begin{aligned}
               & x(n-1)\in \mathcal{C} \implies  x(t)\in \mathcal{C},~\forall t\geq n, \\
               & x(n-1)\notin \mathcal{C} \implies  \lim_{\tau \to \infty} \| x(\tau) \|_{\mathcal{C}} \to 0.
           \end{aligned}
       \end{equation}
\end{thm}

\begin{proof}

    We will show that at each time $t\geq n-1$, $u(t)$ can be found to satisfy the CBF condition for all plausible states:
\begin{equation} \label{eq:cbf_cond_naive}
    H(Ax+Bu(t)) + q \geq (1-\gamma) ( Hx + q),~\forall x\in \myset{X}^{t}.
\end{equation}
Using \eqref{eq:general X0} and Corollary \ref{col:X_t set} to conservatively bound $\myset{X}^t$, it is sufficient to show that $u(t)$ satisfies the following condition:
\begin{multline} \label{eq:cbf_condition_offline1}
    HBu(t)+ H(A - (1-\gamma)I)(x + A^{t}\cup_{\Lambda\in \mathbb{C}_p^{p-2s}} \ker(\mathcal{O}_{\Lambda})) \\+ \gamma q \subseteq \mathbb{R}^l_{\geq 0} 
\end{multline}
for an arbitrary $x\in \myset{X}^{t}$. Under Condition i, we know for any $\Lambda\in \mathbb{C}_{p}^{p-s}$, any $v\in \ker(\mathcal{O}_{\Lambda})$, $Hv = HAv = ... = HA^{n-1} v = 0$. Using the  Cayley–Hamilton Theorem \cite[Theorem 6.1]{hespanha2018linear}, we further deduce that $HA^\tau v=0$ for any $\tau\geq n$. Thus, the condition simplifies to:
\begin{equation*}
\begin{aligned} 
     HBu(t)+ H(A - (1-\gamma)I)x  + \gamma q \geq 0
\end{aligned}
\end{equation*}
for an arbitrary $x\in \myset{X}^{t}$ (as opposed to for all $x\in \myset{X}^{t}$ as in \eqref{eq:cbf_cond_naive}). The existence of such an input $u$ is guaranteed by Condition ii. Thus, the feedback map $k_t$ can be given by, for example, solving the following quadratic program
\begin{equation*}
    \begin{aligned}
        & k_t( \mathcal{D}_t) =  \argmin_u \| u\|^2 \\
        & \hspace{10mm} \textup{s.t. } HBu+ H(A - (1-\gamma)I)x  + \gamma q \geq 0,
    \end{aligned}
\end{equation*}
where the state $x$ is an arbitrary plausible state in $\myset{X}^{t}$ given the input-output data $\mathcal{D}_t$. 
The properties in \eqref{eq:safety_guarantee} thus follow from CBF theory\cite{Ames2017,agrawal2017discrete}.\end{proof}

Theorem~\ref{thm:offline design} provides conditions on $H$ and $q$ such that the system can be rendered safe. An example using these conditions to design a safe set is given later.

We note that while reasoning about system safety offline is useful, the proposed conditions might be too pessimistic in practice since we consider the worst-case attacking scenario. Theoretically, under the $s$-sparse observability condition,  in order to guarantee the system safety starting from $\myset{X}^0$, which is a finite set, we have to take into account the set $\{x_{\textup{true}} \} + \bigcup_{\Lambda\in \mathbb{C}_{p}^{p-2s}} \ker(\mathcal{O}_{\Lambda})$, which is an infinite set in general. Practically, when considering a robotic system, we argue it is not restrictive to assume that the robot starts in a safe region (though precision location may be unknown), safely collects some input-output data (e.g., by staying still at its current position), and {starts to carry out its desired task after enough measurement data acquisition}. In the following subsection, {we presume that measurements of the first few time steps are available and discuss safety conditions therein}.

\subsection{Online phase safe control design}
In this section we derive conditions on the system safety after collecting the input-output data for the first $t\geq n$ steps. Recall that the state uncertainty with the past input-output data is characterized in Proposition \ref{prop:state uncertainty X1}. We note that $s$-sparse observability
is no longer required in this subsection. 

To simplify our analysis, we classify different types of combinations $\Gamma$ based on the dimensions of $\myset{X}^{0,\Gamma}$. Specifically, let 
$$\mathbb{C}_{p}^{p-s} = \mathbb{C}_{\emptyset} \cup \mathbb{C}_{pt} \cup  \mathbb{C}_{sb},$$ where $\mathbb{C}_{\emptyset}, \mathbb{C}_{pt},  \mathbb{C}_{sb} $ are the sets of combinations that correspond to $\myset{X}^{0,\Gamma}$ being an empty set, a singleton, and an affine subspace in $\mathbb{R}^n$, respectively. Let $x^{0,\Gamma}$ 
 satisfy $\mathcal{O}_{\Gamma} x^{0,\Gamma}  = Y_{\Gamma}$ for $\Gamma \in \mathbb{C}_{pt} \cup  \mathbb{C}_{sb}$. Then, from \eqref{eq:X_t_Gamma}, we have 
\begin{equation} \label{eq:explicit x_t_Gamma}
\begin{aligned}
       &  \myset{X}^{t,\Gamma} = \{x^{t,\Gamma} \} & \text{for } \Gamma \in \mathbb{C}_{pt}, \\
       & \myset{X}^{t,\Gamma} = \{x^{t,\Gamma}\}  + A^t (\ker(\mathcal{O}_{\Gamma}))  & \text{for } \Gamma \in \mathbb{C}_{sb},
\end{aligned}
\end{equation}
with $x^{t,\Gamma} = A^t x^{0,\Gamma} + A^{t-1}Bu(0) + A^{t-2}Bu(1) + ... + Bu(t-1)  $ propagated from $x^{0,\Gamma}$.

\begin{lemma} \label{lem:set inclusion}
The set
 $\myset{X}^t \subseteq \myset{C}$ is contained within the safe set if and only if the following conditions hold:
 \begin{enumerate}[i.]
     \item $\ker(\mathcal{O}_{\Gamma})\subseteq \ker(H A^t )$ for all $\Gamma \in  \mathbb{C}_{sb}$;
     \item there exists $x^{t,\Gamma}\in \myset{X}^{t,\Gamma}$ satisfying $Hx^{t,\Gamma} +q \geq 0$ for each $\Gamma \in  \mathbb{C}_{pt} \cup \mathbb{C}_{sb}$.
 \end{enumerate}
\end{lemma}

\begin{proof}
   \textit{Necessity:}  
   Suppose there exists $\Gamma\in\mathbb{C}_{sb}$ such that
   $\ker(\mathcal{O}_{\Gamma}) \not\subseteq \ker(H A^t )$, then there exists $v\in\ker(\mathcal{O}_{\Gamma})$ such that $HA^tv \neq 0$. However, in view of \eqref{eq:explicit x_t_Gamma}, $x(t)+kA^tv\in\myset{X}^t$, for any $k\in\mathbb{R}$, is a plausible state based on data, and a $k$ exists such that $[kHA^tv+Hx(t)+q]_i \not\geq 0$ for some row $i$, i.e.,
   $\myset{X}^t \not\subseteq \myset{C}$. This contradiction proves $\ker(\mathcal{O}_{\Gamma})\subseteq \ker(H A^t )$ for all $\Gamma\in\mathbb{C}_{sb}$. The second condition follows directly from the definition of the safe set.

   \textit{Sufficiency:} 
   We must show that all $x\in\myset{X}^{t,\Gamma}$ for which $x \neq x^{t,\Gamma}$ are also in the safe set, for $\Gamma\in\mathbb{C}_{sb}$ when $\myset{X}^{t,\Gamma}$ is not a singleton.
   For such a state $x$, one calculates that $h(x) - h(x^{t,\Gamma}) = H(x-x^{t,\Gamma})=HA^tv$ for some $v\in \ker(\mathcal{O}_\Gamma)$, using \eqref{eq:explicit x_t_Gamma}. Based on the first condition, $h(x) - h(x^{t,\Gamma})= 0$. Thus $x$ also belongs to $\myset{C}$. This completes the proof.
\end{proof}

We are now in place to derive the safety condition on $u(t)$. The naive discrete-time CBF condition on $u$ is
\begin{equation} \label{eq:cbf_cond_naive2}
    H(Ax+Bu(t)) + q \geq (1-\gamma) ( Hx + q), \forall x\in \myset{X}^{t},
\end{equation}
which may consist of infinitely many linear constraints, making it intractable to solve using the standard quadratic-program-based feedback design.

\begin{thm} \label{thm:cbf_condition}
    Under the premises of Proposition \ref{prop:state uncertainty X1}, the CBF condition in \eqref{eq:cbf_cond_naive2} is equivalent to:
    \begin{align}
    & \ker(\mathcal{O}_{\Gamma}) \subseteq \ker(HA^{t+1} - (1-\gamma)HA^t),~\forall \Gamma \in \mathbb{C}_{sb},  \label{eq:cbf_nullspace_cond}
    \end{align}
    and the existence of a plausible state $x^{t,\Gamma}\in\myset{X}^{t,\Gamma}$ such that there exists an input $u\in \mathbb{R}^m$ that satisfies:
    \begin{align}
    & H(Ax^{t,\Gamma}+Bu) + q \geq (1-\gamma) ( Hx^{t,\Gamma} + q),  
    \label{eq:cbf_multi_cond}
    \end{align}
    for each $\Gamma \in \mathbb{C}_{pt}\cup \mathbb{C}_{sb}$.
    Moreover, when $\myset{X}^0, ..., \myset{X}^{n-1} \subseteq \myset{C}$, the condition in \eqref{eq:cbf_nullspace_cond} holds. 
\end{thm}
\begin{proof}
    In view of Corollary \ref{col:X_t set} and \eqref{eq:explicit x_t_Gamma}, we rewrite the condition in \eqref{eq:cbf_cond_naive2} as
\begin{subequations} \label{eq:CBF_condition_intermediate}
\begin{align} 
    & HBu(t)+ H(A - (1-\gamma)I)x + \gamma q \geq 0,  \nonumber\\
    & \hspace{4,5cm}  \forall x\in \myset{X}^{t,\Gamma}, \Gamma\in \mathbb{C}_{pt}. \label{eq:CBF_condition_intermediate_pt}\\
    & HBu(t)+ H(A - (1-\gamma)I)\myset{X}^{t,\Gamma} + \gamma q \subseteq \mathbb{R}^l_{\geq 0},  \nonumber\\
    & \hspace{6cm} \forall \Gamma\in \mathbb{C}_{sb}.\label{eq:CBF_condition_intermediate_sb}
\end{align}
\end{subequations}
 For $\Gamma\in\mathbb{C}_{pt}$, the set $\myset{X}^{t,\Gamma}$ is a singleton, and   \eqref{eq:CBF_condition_intermediate_pt} is equivalent to \eqref{eq:cbf_multi_cond}. 
 For $\Gamma \in \mathbb{C}_{sb}$,
 one can verify that \eqref{eq:CBF_condition_intermediate_sb} holds if and only if conditions \eqref{eq:cbf_nullspace_cond}
 and \eqref{eq:cbf_multi_cond} hold by following a proof similar to that of Lemma \ref{lem:set inclusion}.

From Lemma \ref{lem:set inclusion}, $\myset{X}^0, ..., \myset{X}^{n-1} \subseteq \myset{C}$ implies $ \ker(\mathcal{O}_{\Gamma})\subseteq \ker(H A^\tau )$ for $\tau = 0,1,...,n-1$ and for any $ \Gamma \in  \mathbb{C}_{sb}$, i.e.,  $Hv = HAv = ... = HA^{n-1}v = 0$ for any $v\in \ker(\mathcal{O}_{\Gamma})$.
From Cayley–Hamilton Theorem \cite[Theorem 6.1]{hespanha2018linear}, we further know $ HA^{\tau}v = 0$ for $\tau \geq n$. Thus, the condition in \eqref{eq:cbf_nullspace_cond} is satisfied.
\end{proof}

\begin{remark}
   The condition on $u(t)$ in \eqref{eq:cbf_multi_cond} is finite and linear, and thus can be implemented as constraints in quadratic program-based control designs.
\end{remark}

We now have the following results regarding system safety.

\begin{thm} \label{thm:safety guarantee_with data}
Under the premises of Proposition \ref{prop:state uncertainty X1}, if $\myset{X}^{0},\dots,\myset{X}^{n-1}\subseteq \myset{C}$, and if the condition in \eqref{eq:cbf_multi_cond} holds at each time step $\tau \geq t$ with a control sequence $\{u(\tau)\}_{\tau\geq t}$,
then the two implications in \eqref{eq:safety_guarantee} hold for the system.
\end{thm}

\begin{proof}
From Theorem \ref{thm:cbf_condition}, we know that the CBF condition \eqref{eq:cbf_cond_naive2} is fulfilled for all states in $\myset{X}^\tau$ for all $\tau\geq t$. As the CBF condition is fulfilled along the system trajectory, we thus know that the implications in \eqref{eq:safety_guarantee} hold following CBF theory \cite{Ames2017,agrawal2017discrete}.
\end{proof}

{Conditions~\eqref{eq:cbf_nullspace_cond} and~\eqref{eq:cbf_multi_cond} in Theorems~\ref{thm:cbf_condition} are equivalent to the naive CBF condition in \eqref{eq:cbf_cond_naive2}. These conditions however can be difficult to verify in certain cases.} Here, we provide alternative, albeit only sufficient, conditions that are easier to verify.

\begin{corollary}
    Sufficient conditions for the existence of an input $u\in\mathbb{R}^m$ such that conditions \eqref{eq:cbf_nullspace_cond} and \eqref{eq:cbf_multi_cond} hold are:
\begin{align}
    & \forall \Gamma \in \mathbb{C}_{sb},~\exists M_\Gamma\in\mathbb{R}^{l\times(p-s)} \text{ s.t. } H = M_\Gamma C_{\Gamma}, \label{eq:CBF_cond_1} \\
    \textup{and } &  HB\in \mathbb{R}^{l\times n} \textup{ is full row rank}. \label{eq:CBF_cond_2}
\end{align}
\end{corollary}
\begin{proof}
     We first show that \eqref{eq:CBF_cond_1} implies condition \eqref{eq:cbf_nullspace_cond}. For any $v\in \ker(\mathcal{O}_{\Gamma})$, we have $\mathcal{O}_{\Gamma}v= 0$ implies $C_\Gamma v= C_\Gamma Av=\dots = C_\Gamma A^{n-1}v = 0$ by definition. Using Cayley-Hamilton Theorem \cite[Theorem 6.1]{hespanha2018linear}, we further deduce $C_\Gamma A^{t+1}v = C_\Gamma A^t v =0$.  Therefore, $(HA^{t+1} - (1-\gamma)HA^t)v = M_\Gamma C_{\Gamma} (A^{t+1} - (1-\gamma)A^t )v = 0$.
     That is, $v\in \ker(HA^{t+1} - (1-\gamma)HA^t)$. Thus the equality in \eqref{eq:cbf_nullspace_cond} holds for any $\Gamma \in \mathbb{C}_{sb}$. 

     We now show that \eqref{eq:CBF_cond_2} implies the existence of an input $u$ such that \eqref{eq:cbf_multi_cond} holds. Note that there are finite number of sets $\Gamma\in \mathbb{C}_{pt}\cup \mathbb{C}_{sb}$, i.e. $|\mathbb{C}_{pt}\cup \mathbb{C}_{sb}|=k\in\mathbb{N}$, so we seek to show there exists $u\in\mathbb{R}^m$ such that $HBu + z \geq 0$,
     where $z = (z_1, ..., z_k)$ and 
     \begin{multline*}
         z_i := \min_{\Gamma \in \mathbb{C}_{pt}\cup \mathbb{C}_{sb}}[H(A - (1-\gamma)I)A^t x^{t,\Gamma}  + \gamma q ]_{i}
     \end{multline*}
When \eqref{eq:CBF_cond_2} holds, one feasible solution to \eqref{eq:cbf_multi_cond} is  thus $u = B^\top H^\top (HB B^\top H^\top)^{-1} z$.
\end{proof}

\section{A $2$D-vehicle example}
Consider an omnidirectional vehicle on a $2$D-plane with continuous-time dynamics
\begin{equation}
    \dot{x} = \begin{bsmallmatrix}
        0 & 1 & 0 & 0 \\
        0 & -0.2 & 0 & 0 \\
        0 & 0 & 0 & 1 \\
        0 & 0 & 0 & -0.2
    \end{bsmallmatrix} x + \begin{bsmallmatrix}
        0 & 0 \\
        1 & 0 \\
        0 & 0 \\
        0 & 1 
    \end{bsmallmatrix}u
\end{equation}
where $x=(x_1,x_2,x_3,x_4), u = (u_1, u_2)$, $x_1,x_3$ are the $x,y$-axes positions, $x_2, x_4$ the respective velocities, and $u_1,u_2$ the respective acceleration-level inputs. For digital implementation of our controller, this continuous-time system is discretized in time using zero-order hold method with sampling time $0.01$s. The system output is given by $y_i = x_{\lceil i/2 \rceil} + e_i$, for $i = 1,2,\ldots, 8$,
where recall $e =  (e_1, e_2, \ldots, e_8)$ represents the attacking signal.

\vspace{2mm}
\noindent \textit{Offline safety guarantee:} One verifies that this system is $1$-sparse observable. When only $1$ sensor attack exists, following Proposition \ref{prop:general sparse case}, we know the set of plausible initial states $$\myset{X}^0 \subseteq  \{x_{\textup{true}}\} +  \text{span}\begin{bsmallmatrix}
    1\\
    0\\
    0\\
    0
\end{bsmallmatrix}  \cup \text{span}\begin{bsmallmatrix}
    0\\
    0\\
    1\\
    0
\end{bsmallmatrix}.$$ From Theorem 
\ref{thm:offline design}, one possible safe region is given by $\myset{C} = \{x: h(x) = H x + q \geq 0\}, $ where $H = M \begin{bsmallmatrix}
    0 & 1 & 0 & 0 \\
    0 & 0 & 0 & 1
\end{bsmallmatrix}$ for any $M$ with proper dimensions and $q$ such that the function $h$ is a CBF. Theorem \ref{thm:offline design} guarantees that if the system state is inside the safe set at the 4th step, then the system can be rendered safe in all possible attacking scenarios. One example choice is $\myset{C} = \{x\in \mathbb{R}^4: -4 \leq x_i\leq 4, i = 2,4\}.$

\vspace{2mm}
\noindent \textit{Online safe control:} In this case, we choose a hyper-rectangle safe set as $    \myset{C} = \{x\in \mathbb{R}^4: Hx + q\geq 0\}$ with $H = \begin{bsmallmatrix}
    I_{4\times 4} \\
    -I_{4\times 4}
\end{bsmallmatrix}$ and $q = 4\mathbf{1}_8$.
Here, sensors $1,3,5$ are under attack and {the attacker intends to drive the system out of the safe set. To achieve this target and also stay stealthy, one attacking strategy is to pick a fake initial state $x_{\textup{fake}} =  (2,2,2,1) $ and generate the attacking signal based on this fake state and consistent with system dynamics. Arguably, this is the best strategy to avoid attack detection.}
In this simulation, the actual initial state $x_{\textup{true}} = (1,1,1,1)$. Gaussian noises from $\mathcal{N}(0,0.01^2)$ are added to all sensor measurements.

For online state estimation, we take a brute-force approach by looking at all possible sensor combinations in $\mathbb{C}^{p-s}_p$ and determining whether the sensors are intact. This is done by computing the least squares solution to the linear equality $\mathcal{O}_{\Gamma}x = Y_{\Gamma}$, and empirically checking the bound of the matching error. In this simulation, we set the error threshold to $0.001$. We note that even though in the theoretical development we consider all the measurements from time step $0$ to $t$ for determining the state at time step $t$, it is not necessary nor practical to store all those measurements for secure state reconstruction. Instead, we will use measurements at $t-3, t-2, t-1, t$ to determine the plausible states at time step $t$. The plausible states at time step $t-3$ obtained from \eqref{eq:X0set} are then propagated to the plausible states at time step $t$ according to \eqref{eq:X_t_Gamma}.

\setlength{\belowcaptionskip}{-15pt}

 \begin{figure}[t]
     \centering
     \includegraphics[width=0.8\linewidth]{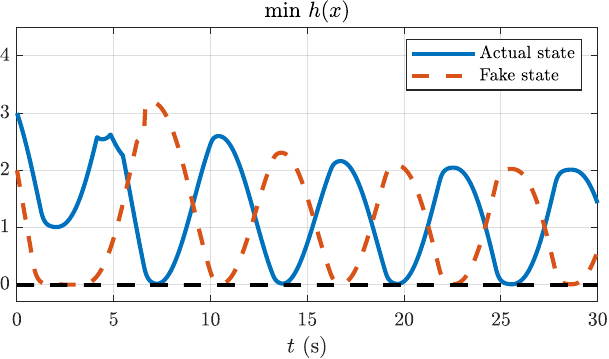}
     \caption{CBF evolution over time}
     \label{fig:cbf}
 \end{figure}
 \begin{figure}[t]
     \centering
     \includegraphics[width=0.8\linewidth]{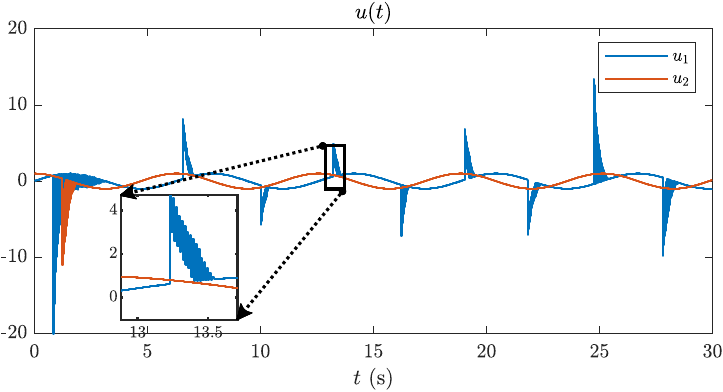}
     \caption{Safe control input over time}
     \label{fig:input}
 \end{figure}

 According to Theorem \ref{thm:cbf_condition}, we apply the following online safe controller 
\begin{equation*} \label{eq:qp_control}
    \begin{aligned}
            u(t)  &= \argmin_u \| u - u_{\text{nom}} \| \\
            & \text{s.t. \eqref{eq:cbf_multi_cond} holds,} \quad  \forall \ \Gamma \in \mathbb{C}_{pt}\cup \mathbb{C}_{sb} 
    \end{aligned}
\end{equation*}
together with a one-time checking mechanism on whether the plausible states at the first $4$ steps lie within the safe set $\myset{C}$. From Theorem \ref{thm:safety guarantee_with data}, system safety is guaranteed if this QP is always feasible. A nominal input at time step $\tau$ is chosen as $u_{\text{nom}}(\tau) = (\sin(0.01\tau),\cos(0.01\tau))$. This control signal is implemented for the first three steps and then filtered out safely by the QP, where $\gamma$ is chosen to be $0.05$.

 During implementation, the online state estimator finds $4$ plausible initial states
 $(1,1,1,1),~(1,1,2,1),$ $~(2,2,1,1),~(2,2,2,1)$. We note that the system cannot distinguish the actual state from the other three fake states.

 \begin{figure}[t]
    \centering
    \includegraphics[width=1.0\linewidth]{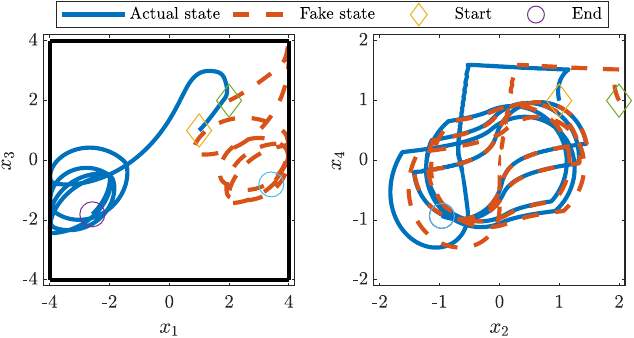}
    \caption{{Actual trajectory (blue solid line) v.s. fake trajectory (dash red line). The black square in the left subfigure represents the safe region for $x,y$-axes positions.}}
    \label{fig:traj}
\end{figure}
We observe from Fig. \ref{fig:cbf} that both the actual trajectory and the fake trajectory respect the safety constraint for all time. This, however, is not always the case. For example, when the attacker chooses $x_{\textup{fake}} =  (2,2,2,2)$, the corresponding fake trajectory does not fulfill the safety requirement. This is due to the fact that by attacking sensors $1, 3,$ and $5$, the attacker can not confuse the system about the velocity along $y$-axis $x_4$, and we can determine that $x_4 = 1 $ for all the plausible initial states $x\in \myset{X}^0$.

Figure \ref{fig:input} demonstrates that the safe control input closely resembles the nominal one, and only modifies the nominal control when necessary. One can also see from Fig. \ref{fig:traj} that although the velocity responses of the actual and the fake trajectories share some similarities, the positional movement is very different. Our proposed safe controller correctly constrains all possible trajectories inside the safe region (the square enclosed by black lines).

{
   We acknowledge that while the proposed solution works in this example, it may not scale up to large systems due to online computational burden. The main bottleneck lies in the combinatorial nature of secure state reconstruction.  
   We plan to reduce online computational costs in our future work.
}
\section{Conclusions}
In this work, we have provided conditions that guarantee safety for discrete-time linear systems under severe sensor attacks. We consider a scenario where the secure state reconstruction problem may have non-unique solutions. We first provide an exact characterization of these plausible states. We then derive conditions for designing an offline safe set, which guarantees system safety for all possible sensor attacks under a mild sparse observability condition. When the system is deployed online with measurement data for the first few steps available, we have proposed a quadratic program-based safe control scheme. We show that system safety hinges on a finite number of control barrier function conditions, and on a kernel condition that is related to the system dynamics and the number of attacked sensors. A numerical example of a $2$-D vehicle illustrates theoretical results. 

\bibliographystyle{IEEEtran}
\bibliography{IEEEabrv,references}

\end{document}